\documentclass[12pt,twoside]{article}

 \usepackage{float}
\usepackage{graphicx}
\usepackage{epstopdf}
\usepackage{graphicx}
\usepackage{epic}
\usepackage{multirow}
\usepackage{tikz}
\usepackage{threeparttable}
\usepackage{xcolor}
\usepackage{rotating}
\usepackage{nicematrix}
\usepackage{tikz}

\usetikzlibrary{arrows,shapes,chains}

\renewcommand{\paragraph}{\roman{paragraph}}
\usepackage[a4paper]{geometry}
\usepackage{geometry}
\makeatletter
\renewcommand\title[1]{\gdef\@title{\reset@font\Large\bfseries #1}}
\renewcommand\section{\@startsection {section}{1}{\z@}%
                                   {-3.5ex \@plus -1ex \@minus -.2ex}%
                                   {2.3ex \@plus.2ex}%
                                   {\normalfont\large\bfseries}}
\renewcommand\subsection{\@startsection{subsection}{2}{\z@}%
                                     {-3ex\@plus -1ex \@minus -.2ex}%
                                     {1.5ex \@plus .2ex}%
                                     {\normalfont\normalsize\bfseries}}
\renewcommand\subsubsection{\@startsection{subsubsection}{3}{\z@}%
                                     {-2.5ex\@plus -1ex \@minus -.2ex}%
                                     {1.5ex \@plus .2ex}%
                                     {\normalfont\normalsize\bfseries}}

\def\@runningauthor{}\newcommand{\runningauthor}[1]{\def\runningauthor{#1}}
\def\@runningtitle{}\newcommand{\runningtitle}[1]{\def\runningtitle{#1}}

\renewcommand{\ps@plain}{%
\renewcommand{\@evenhead}{\footnotesize\scshape \hfill\runningauthor\hfill}
\renewcommand{\@oddhead}{\footnotesize\scshape \hfill\runningtitle\hfill}}

\g@addto@macro\bfseries{\boldmath}

\makeatother

\usepackage{amsthm,amsmath,amssymb}
\usepackage{cite}
\usepackage{graphicx}

\usepackage[colorlinks=true,citecolor=black,linkcolor=black,urlcolor=blue]{hyperref}

\theoremstyle{plain}
\newtheorem{theorem}{Theorem}[section]

\newtheorem{lemma}[theorem]{Lemma}

\theoremstyle{definition}
\newtheorem{definition}[theorem]{Definition}
\newtheorem{example}[theorem]{Example}

\newtheorem{problem}[theorem]{Problem}

\theoremstyle{remark}
\newtheorem{remark}[theorem]{Remark}

\runningauthor{}

\date{}

\begin{document}

\title{New characterization of full weight spectrum one-orbit cyclic subspace codes\thanks{This research is supported by the National Natural Science Foundation of China under Grant 12471490. {\em (Corresponding author: Minjia Shi.)}}}
\author{Minjia Shi\thanks{Minjia Shi is with the Key Laboratory of Intelligent Computing and Signal Processing, Ministry of Education, School of Mathematical Sciences, Anhui University, Hefei, China (email: smjwcl.good@163.com).},~
Wenhao Song\thanks{Wenhao Song is with the School of Computer Science and Technology, Anhui University, Hefei, China (email: s1114967601@gmail.com).}}

\date{}
    \maketitle

\begin{abstract}
Castello $\textit{et al}$. [J. Comb. Theory Ser. A, 212, 106005 (2025)] provided a complete classification for full weight spectrum (FWS) one-orbit cyclic subspace codes. In this paper, we determine the weight distributions of a family of FWS codes and exhibit some equivalence classes of FWS codes under certain conditions. Furthermore, we provide a complete classification for $r$-FWS codes.
\end{abstract}
{\bf Keywords:} Cyclic subspace code, full weight spectrum code, $r$-full weight spectrum code \\
\noindent{\bf Mathematics Subject Classification}

\section{Introduction}\label{pre}

Subspace codes, introduced by Koetter and Kschischang \cite{random}, are collections of subspaces of a finite vector space \(\mathbb{F}_q^n\), designed for error control in random network coding. In contrast to classical linear codes under the Hamming metric, subspace codes are measured by the \textit{subspace distance}.
Let $k,\ n \in \mathbb{Z}$ satisfy $0 \leq k \leq n$. The {\em Grassmannian space} over $\mathbb{F}_q$, denoted by $\mathcal{G}_q(n,k) \subseteq \mathcal{P}_q(n)$, consists of all $k$-dimensional $\mathbb{F}_q$-subspaces of the vector space $\mathbb{F}_q^n$, where $\mathcal{P}_q(n)$ represents the collection of all $\mathbb{F}_q$-subspaces of $\mathbb{F}_q^n$. A \emph{constant dimension subspace code} refers to any subset $\mathcal{C} \subseteq \mathcal{G}_q(n,k)$ equipped with the subspace metric defined by
\[
d(U,V) = 2k - 2\dim_{\mathbb{F}_q}(U \cap V) \quad  \ U,\ V \in \mathcal{C}.
\]
The {\em minimum distance} of such a code is given by
\[
d(\mathcal{C}) = \min\{ d(U,V) \mid U,V \in \mathcal{C},\ U \neq V \}.
\]
This distance captures the packet loss and mixing phenomena typical in network transmissions, making subspace codes a more natural framework in such contexts.

Among subspace codes, constant-dimension codes (CDCs), in which all codewords are \(k\)-dimensional subspaces, are of particular interest due to their strong algebraic and geometric properties. 
A prominent family of CDCs is that of \textit{cyclic orbit codes}. Specifically, a code $\mathcal{C} \subseteq \mathcal{G}_q(n,k)$ is called \emph{cyclic} if it satisfies the closure property: $\alpha V \in \mathcal{C}$ for all $\alpha \in \mathbb{F}_{q^n}^*$ and $V \in \mathcal{C}$. When $\mathcal{C}$ coincides with a single orbit $\text{Orb}(S) = \{ \alpha S \mid \alpha \in \mathbb{F}_{q^n}^* \}$ for some subspace $S \subseteq \mathbb{F}_{q^n}$, it is termed a \emph{one-orbit cyclic subspace code} with $S$ as its representative. 
Similar to the classical Hamming case, the \textbf{(distance) weight distribution} of \( \mathcal{C} \) is defined by \( (\omega_2( \mathcal{C}), \ldots, \omega_{2k}( \mathcal{C})) \) (see \cite{distance}), where
	\[
	\omega_{2i}( \mathcal{C}) = \left| \left\{ \alpha S \in \text{Orb}(S) : \alpha \in \mathbb{F}_{q^n}^* , d(S, \alpha S) = 2i \right\} \right| 
	\]
	for \( i \in \{1, \ldots, k\} \).
Note that the definition of weight distribution of $ \mathcal{C} = \text{Orb(S)}$ does not depend on the
choice of $S$.

To formalize our goals, let us define the quantity \(\mathcal{L}(n,k,q)\) to be the maximum number of distinct nonzero weights attainable by one-orbit cyclic subspace codes in \(\mathcal{G}_q(n,k)\). This number equivalently captures the maximal number of distinct pairwise distances in such codes. For a one-orbit code \(\mathcal{C} = \text{Orb}(S)\), the weight spectrum is determined by the intersection dimensions \(\dim_{\mathbb{F}_q}(S \cap \alpha S)\) as \(\alpha\) ranges over \(\mathbb{F}_{q^n}^*\). In particular, the distance values satisfy
\[
2 \leq d(S, \alpha S) \leq 2k,
\]
for all \(\alpha \in \mathbb{F}_{q^n}^*\) with \(S \neq \alpha S\), and since distances are always even integers, it follows that
\[
\mathcal{L}(n, k, q) \leq k.
\]
In analogy with the Hamming case, we call a code with exactly \(k\) distinct nonzero distances a \textit{full weight spectrum} (FWS) code. Generally, a code is an $r$-FWS code if the last $r$ entries of the weight distribution are zeroes and all the others are nonzero. In particular, $0$-FWS codes correspond to FWS codes.

The theory of cyclic subspace codes, particularly motivated by applications in random network coding, was fundamentally shaped by the pioneering work of Etzion and Vardy \cite{projective}, whose groundbreaking projective space framework first revealed the cyclic orbit structure and formalized error correction in projective space. Their framework not only established essential bounds and constructions, but also highlighted the combinatorial richness of the Grassmannian space, sparking a wide range of subsequent research. Among these developments, cyclic orbit codes have emerged as a structured and algebraically elegant class of subspace codes, attracting more attention in recent years \cite{cyclic3, cyclic4}. 

To further enhance the algebraic understanding and construction of cyclic orbit codes, researchers have introduced new tools rooted in linearized polynomial algebra. In particular, \textit{subspace polynomials}, as proposed by Etzion et al.~\cite{subpoly} offer a powerful representation of subspaces in the Grassmannian space by associating them with special classes of linearized polynomials. This framework not only unifies various orbit code constructions, but also facilitates explicit analysis of parameters such as orbit lengths and minimum distances. Building on this approach, Zhao and Tang \cite{cyclic4} extended the constructions to broader classes using generalized subspace polynomials, yielding codes with size up to \(\frac{q^N - 1}{q - 1}\) and minimum distance \(2k - 2\). Trautmann et al.~\cite{equivalence} provided a classification of cyclic orbit codes and a decoding procedure for specific subclasses. Recently, Gluesing-Luerssen and Lehmann \cite{distance} studied the weight distribution of cyclic orbit codes, highlighting its potential as a tool for further classification. 

Beyond their structural elegance, subspace codes, and particularly orbit codes, have found impactful applications in \textit{vector network coding}. Etzion and Wachter-Zeh \cite{etzion2016vector} demonstrated that vector network coding, utilizing subspace codes derived from rank-metric constructions, can substantially reduce the required field size in multicast networks compared to traditional scalar linear schemes. Their work not only provides a profound theoretical breakthrough in bridging coding theory and network coding but also offers practical coding constructions that pave the way for more efficient and scalable network communication protocols.

Recently, Castello et al. \cite{jcta} studied the interplay between subspace codes and FWS codes by focusing on one-orbit cyclic subspace codes with minimum subspace distance two. They provided a complete classification of such codes by identifying two infinite families of cyclic orbit codes constructed via polynomial bases as follows. 

\begin{theorem}\cite[Theorem 1.2]{jcta}\label{Th1.2}
	Let $\mathcal{C}$ be a one-orbit cyclic orbit code in $\mathcal{G}_q(n, k)$. Then $\mathcal{C}$ is a full weight spectrum code if and only if $\mathcal{C} = \operatorname{Orb}(S)$, where $S$ is one of the following:
	
	\begin{enumerate}
		\item[$(1)$] $S = \left\langle 1, \lambda, \ldots, \lambda^{k-1} \right\rangle _{\mathbb{F}_q}$ for some $\lambda \in \mathbb{F}_{q^n} \setminus \mathbb{F}_q$, where
		\[
		k \leq \begin{cases}
			\frac{ \left[ \mathbb{F}_q(\lambda)^2 :~\mathbb{F}_q \right] + 1 }{2}& \text{if } \dim_{\mathbb{F}_q} (\mathbb{F}_q(\lambda)) < n, \\
			\frac{n}{2} & \text{if } \dim_{\mathbb{F}_q} (\mathbb{F}_q(\lambda)) = n,
		\end{cases}
		\]
		
		\item[$(2)$] $S = \left\langle 1, \lambda, \ldots, \lambda^{l-1} \right\rangle_{\mathbb{F}_{q^2}} \oplus \lambda^l \mathbb{F}_q$ for some $\lambda \in \mathbb{F}_{q^n} \setminus \mathbb{F}_{q^2}$, where $k = 2l + 1$, $n$ is even, and $l < \frac{[ \mathbb{F}_{q^2}(\lambda) :~\mathbb{F}_{q^2} ]}{2}$.
	\end{enumerate}

\end{theorem}

Their results highlighted a promising connection between cyclic subspace codes and classical FWS codes. In particular, they determined the weight distribution of the first family of FWS codes in Theorem \ref{Th1.2} and left the weight distribution of the second family of codes as an open question. 

\begin{problem}\label{problem1}
It would be nice to determine the weight distribution of the codes in (2) of Theorem \ref{Th1.2}. Note that the weight distributions of the codes in (1) and those in (2) of Theorem \ref{Th1.2} are different.
\end{problem}

Furthermore, they also proposed two related problems:

\begin{problem}\label{problem2}
Is it possible to determine the equivalence classes of the codes in Theorem \ref{Th1.2} under the action of linear isometries?
\end{problem}

\begin{problem}\label{problem3}
Is it possible to characterize $r$-FWS codes similarly to what has been done for FWS codes?
\end{problem}

This paper aims to fill these gaps. We determine the weight distribution of the codes in (2) of Theorem \ref{Th1.2}. We also investigate code equivalence in the context of the normalizer \(N_{GL_n(q)}(\mathbb{F}_{q^n}^*)\), offering partial progress on the classification. Finally, we study the existence of \(r\)-FWS codes and provide a structural characterization of when such codes can exist.

The paper is organized as follows. In Section \ref{sec3}, we focus on Problem \ref{problem1}. In Section \ref{sec4}, we exhibit some equivalence classes by studying the case when the linear map belongs to \(N_{GL_n(q)}(\mathbb{F}_{q^n}^*)\). In Section \ref{sec5}, we  characterize the \(r\)-FWS codes. As a consequence, only in some special cases do we find that \(r\)-FWS codes exist. In Section \ref{sec7}, we summarize and conclude our paper.

\section{The weight distribution of the second family of codes in Theorem \ref{Th1.2}} \label{sec3}
\noindent In \cite[Proposition 4.10]{jcta}, Castello \textit{et al}. analyzed the classification of elements $\mu\in \mathbb{F}_{q^n}$ to derive the dimension of $S\cap \mu S$ for the second class of codes in Theorem \ref{Th1.2}, yet they did not specify the explicit weight distribution for this class codes. In this section, we rigorously determine the weight distribution of the second class of codes in Theorem \ref{Th1.2}, thereby addressing Problem \ref{problem1}.

Given an $\mathbb{F}_q$-subspace $S$ of $\mathbb{F}_{q^n}$, the \textbf{stabilizer} of $S$ is
\[
H(S) = \{x \in \mathbb{F}_{q^n}^* : xS = S\} \cup \{0\}.
\]
Note that $H(S)$ is a subfield of $\mathbb{F}_{q^n}$ and $S$ is linear over $H(S)$. This implies that also
$S \cap \alpha S$ is linear over $H(S)$ for any $\alpha \in \mathbb{F}_{q^n}$. The following lemma gives us information on the weight distribution of a one-orbit cyclic subspace code in relation to the stabilizer
of one of its representatives.
\begin{lemma}\cite[Lemma 3.1]{jcta}\label{congruennce}
	Let $\mathcal{C} = \text{Orb}(S) \subseteq \mathcal{G}_q(n, k)$ and let $(\omega_2(\mathcal{C}), \ldots, \omega_{2k}(C))$ be its weight distribution. If $\omega_{2i}(\mathcal{C}) > 0$ for some $i \in \{1, \ldots, k\}$, then $k \equiv i \pmod{[{H}(S) : \mathbb{F}_q]}$.
\end{lemma}
\begin{lemma}\cite[Proposition 4.10]{jcta} \label{prop4.10}
	Let \( n \) be a positive even integer, let \( \lambda \in \mathbb{F}_{q^n} \setminus \mathbb{F}_{q^2} \), let \( Y = \langle 1, \lambda, \ldots, \lambda^l \rangle_{\mathbb{F}_{q^2}} \) and \( \overline{S} = \langle 1, \lambda, \ldots, \lambda^{l-1} \rangle_{\mathbb{F}_{q^2}} \) with \( 1 \leq l < \frac{t}{2} \) where \( t = [\mathbb{F}_{q^2}(\lambda) : \mathbb{F}_{q^2}] \).
	If \( S \) is the \( \mathbb{F}_q \)-subspace of \( Y \) given by
	\[ S = \langle 1, \lambda, \ldots, \lambda^{l-1} \rangle_{\mathbb{F}_{q^2}} \oplus \lambda^l \mathbb{F}_q \in \mathcal{G}_q(n, 2l + 1), \]
	then $\dim_{\mathbb{F}_q}(S \cap \mu S) =$
	\[
	\begin{cases}
		2l, & \text{if and only if } \mu \in \mathbb{F}_{q^2} \setminus \mathbb{F}_q, \\
		2(l - r), & \text{if and only if } \mu \text{ (or } \mu^{-1}\text{) is of the form } \frac{p_1(\lambda)}{p_2(\lambda)} \text{ for some } p_1(x), p_2(x) \in \mathbb{F}_{q^2}[x] \\
		& \text{such that } \gcd(p_1(x), p_2(x)) = 1, p_1(x) \text{ monic}, \deg p_1(x) = \deg p_2(x) = r, \\
		&\text{and } c(p_2(x), x^r) \in \mathbb{F}_{q^2} \setminus \mathbb{F}_q. \\
		& \text{In this case } S \cap \mu S = \overline{S} \cap \mu \overline{S}. \\
		2(l - r) + 1, & \text{if and only if } \mu \text{ (or } \mu^{-1}\text{) is of the form } \frac{p_1(\lambda)}{p_2(\lambda)} \text{ for some } p_1(x), p_2(x) \in \mathbb{F}_{q^2}[x] \\
		& \text{such that } \gcd(p_1(x), p_2(x)) = 1, p_1(x) \text{ monic}, \deg p_1(x) = r \geq \deg p_2(x), \\
		&\text{ and } c(p_2(x), x^r) \in \mathbb{F}_q. \\
		& \text{In this case } \dim_{\mathbb{F}_q}(S \cap \mu S) = \dim_{\mathbb{F}_q}(\overline{S} \cap \mu \overline{S}) + 1.
	\end{cases}
	\]
	Finally, for any \( r \in \{1, \ldots, l - 1\} \) (resp. \( r \in \{1, \ldots, l\} \)) there exist elements \( \mu \in \mathbb{F}_{q^n} \) for which \( \dim_{\mathbb{F}_q}(S \cap \mu S) = 2(l - r) \) (resp. \( \dim_{\mathbb{F}_q}(S \cap \mu S) = 2(l - r) + 1 \)).
	
\end{lemma}

As observed in Lemma $\ref{prop4.10}$, to determine the weight distribution, we need to count the pair $(p_1(\lambda),p_2(\lambda))$ with the form $\cfrac{p_1(\lambda)}{p_2(\lambda)}$, where $p_1(\lambda),p_2(\lambda)\in \mathbb{F}_{q^2}[x]$ also satisfies some properties, which is closely related to the enumeration of coprime polynomials, hence we need the following result.
\begin{lemma}\cite[Theorem 3]{coprimepolynomials}\label{probility}
	Let $\mathbb{F}$ be a finite field of $q$ elements, and let $a(x)$ and $b(x)$ be randomly chosen from the set of polynomials in $\mathbb{F}[x]$ of degree $m$ and $n$, respectively, where $m$ and $n$ are not both zero. Then the probability that $a(x)$ and $b(x)$ are relatively prime is $1-\cfrac{1}{q}$.
\end{lemma}
We are now in a position to state and prove the main result of this section.
\begin{theorem}\label{thdistribution}
	Let \( n \) be a positive even integer, let \( \lambda \in \mathbb{F}_{q^n} \setminus \mathbb{F}_{q^2} \), let \( Y = \langle 1, \lambda, \ldots, \lambda^l \rangle_{\mathbb{F}_{q^2}} \) and \( 1 \leq l < \frac{t}{2}\) where \( t = [\mathbb{F}_{q^2}(\lambda) : \mathbb{F}_{q^2}] \). If \( S \) is the \( \mathbb{F}_q \)-subspace of \( Y \) given by
	\[
	S = \langle 1, \lambda, \ldots, \lambda^{l-1} \rangle_{\mathbb{F}_{q^2}} \oplus \lambda^l \mathbb{F}_q \in \mathcal{G}_q(n, 2l + 1),
	\]
	then 
	\[
	\omega_{2i} = 
	\begin{cases} 
		q, & \text{if } i = 1, \\
		&\\
		q^{4r-1}(q^2 - 1), & \text{if } i = 2r + 1,\  r \in \{1, \ldots, l - 1\}, \\
		&\\
		q^{4r-2}(q+1)^2, & \text{if } i = 2r,\  r \in \{1, \ldots, l\}, \\
		&\\
		\frac{q^n-1}{q-1}-(q+1)-(q+1)q^2\frac{(q^{4l-3}-q)(q-1)+(q+1)(q^{4l}-1)}{q^4-1} & \text{if}\  i=2l+1. \\
	\end{cases}
	\]
\end{theorem}
\begin{proof}We start by counting the number of $\mu \in \mathbb{F}_{q^n}$ such that $\text{dim}(S\cap \mu S)=2l+1$. 
	In this case, $\mu \in H(S)$. Since $\mathcal{C}=\text{Orb}(S)$ is a FWS code and $\omega_{4l}>0$, we obtain that $H(S)=\mathbb{F}_q$ by Lemma $\ref{congruennce}$. Hence, $\omega_{0}=1$. Moreover, $\mu S=S$ if and only if $\mu \in \mathbb{F}_q^*$. By Lemma $\ref{prop4.10}$, $\text{dim}(S\cap \mu S)=2l$ if and only if $\mu \in \mathbb{F}_{q^2}\setminus \mathbb{F}_q$. Therefore, $\omega_{2}=q$. Let $c = c(p_2(x), x^r)$.
	
	If $\text{dim}_{\mathbb{F}_q}(S \cap \mu S)=2(l-r)$, then
	\[
	\mu =\frac{p_1(\lambda)}{p_2(\lambda)} \iff \mu = c ^{-1}\cdot\frac{p_1(\lambda)}{\widetilde{p_{2}}(\lambda)}
	,\]
	where $c\in \mathbb{F}_{q^2}\setminus \mathbb{F}_q$, $p_1(x), p_2(x) \in \mathbb{F}_{q^2}[x]$, $\gcd(p_1(x), p_2(x)) = 1, p_1(x) \text{ monic}, \deg p_1(x) = \deg p_2(x) = r$, and \( p_2(\lambda) = c \widetilde{p_{2}}(\lambda) \), thus \( \widetilde{p_{2}}(x) \) is monic.
	According to Lemma $\ref{probility}$, the corresponding number of such pairs is 
	\[ q^{4r-2}(q^2 - 1).
	\]
	Moreover, the ordered triples \( (c, p_1(x), p_2(x)) \) correspond bijectively to \( \mu \). Note that the number of the values of $c$ is $q^2-q$. 
	Therefore,
	\begin{equation}\label{2r}
		\begin{aligned}
			& \left| \left\{ \mu S : \mu \in \mathbb{F}_{q^2}(\lambda) \setminus \mathbb{F}_{q^2}, \dim_{\mathbb{F}_q}(S \cap \mu S) = 2(l - r) \right\} \right| \\
			&= (q^2 - q) \cdot q^{4r-2}(q^2 - 1) \cdot \frac{1}{q - 1} \\
			&= q^{4r-1}(q^2 - 1).
		\end{aligned}
	\end{equation}
	Thus, $\omega_{4r+2}=q^{4r-1}(q^2 - 1)$ for \( r \in \{1, \ldots, l - 1\} \).
	
	If $\text{dim}_{\mathbb{F}_q}(S\cap\mu S)=2(l-r)+1$. We assume \( p_2(x) \) is of degree \( s \), then 
	\[
	\mu = \frac{p_1(\lambda)}{p_2(\lambda)} \iff \mu = d^{-1} \cdot \frac{p_1(\lambda)}{\widetilde{p_2}(\lambda)}
	,\]
	where \( c\in \mathbb{F}_q,\ d = c(p_2(x), x^s) \in \mathbb{F}_{q^2} \), $p_1(x), p_2(x) \in \mathbb{F}_{q^2}[x], \gcd(p_1(x), p_2(x)) = 1$, $p_1(x)$ monic, $\deg p_1(x) = r \geq \deg p_2(x)=s$, and \( p_2(\lambda) = d \widetilde{p_{2}}(\lambda) \), thus \( \widetilde{p_{2}}(x) \) is monic. It divides into two cases:\begin{itemize}
		\item [$(1)$] If $c = 0$, then \( s<r \). And we assume that $\mu$ is the form of $\frac{p_1(\lambda)}{p_2(\lambda)}$. Similarly, we need to count all ordered pairs of coprime polynomials \( (p_1(x), p_2(x)) \) satisfying the conditions in this case. By Lemma $\ref{probility}$, if \( \deg p_2(x) = s \geq 1 \), there are
		\[
		(q^2)^r \cdot (q^2)^s \cdot \left(1 - \frac{1}{q^2}\right) = q^{2r + 2s - 2} (q^2 - 1)
		\]
		ordered pairs of coprime monic polynomials over $\mathbb{F}_{q^2}[x]$. Since there are $q^2-1$ choices of $d\in\mathbb{F}_{q^2}^*$, there are $ q^{2r + 2s - 2} (q^2 - 1)^2$ different $\mu$'s in total. If \( \deg p_2(x) = 0 \), there are \( q^{2r} (q^2 - 1) \) ordered pairs in total (since every non-zero constant is coprime with the polynomial). Consequently, the number of the ordered pairs \( (p_1(x), p_2(x)) \) for $c=0$ is
		\begin{align*}
			\sum_{s=1}^{r-1} q^{2r + 2s - 2} (q^2 - 1)^2 + q^{2r} (q^2 - 1) &= q^{2r - 2} (q^2 - 1) \sum_{s=1}^{r-1} q^{2s} + q^{2r} (q^2 - 1) \\
			&= q^{2r - 2} (q^2 - 1)^2 \cdot \frac{q^{2r} - q^2}{q^2 - 1} + q^{2r} (q^2 - 1) \\
			&= q^{4r-2}(q^2-1).
		\end{align*}
		Notice that either $\mu$ or $\mu^{-1}$ can be the form of $\frac{p_1(\lambda)}{p_2(\lambda)}$, so we have $2q^{4r-2}(q^2-1)$ different $\mu$'s in total. 
		\item [$(2)$] If $c\ne 0$, then \( p_2(x) \) is of degree \( r \) and $d=c \in \mathbb{F}_q$. By Eq. $(\ref{2r})$, there are \( q^{4r - 2} (q^2 - 1) \) choices for \( \cfrac{p_1(\lambda)}{\widetilde{p_2}(\lambda)} \) and $q-1$ choices for $d$.
		\end{itemize}

	Similarly, \( \mu S=S \) if and only if $\mu \in \mathbb{F}_q^*$, hence we obtain that
	\begin{equation}\label{2l+1}
		\begin{aligned}
			w_{4r} &= |\{ \mu S \in \text{Orb}(S) : \mu \in \mathbb{F}_q^{*n}, d(S, \mu S) = 4r \}| \\
			&= \frac{1}{q - 1}\left((q-1)q^{4r-2}(q^2-1) + 2q^{4r-2}(q^2-1)\right)\\
			&=q^{4r-2}(q+1)^2
		\end{aligned}
	\end{equation}
	for \( r \in \{1, \ldots, l\} \).
	
	If $\text{dim}(S\cap \mu S)=0$, then \begin{align*}
		\omega_{4l+2}&=\frac{{{q^n} - 1}}{{q - 1}} - q - 1 - \sum_{i = 1}^{l-1}\left(q^{4i-1}(q^2-1)\right)-\sum_{i=1}^{l}\left(q^{4i-2}(q+1)^2\right)\\
		&=\frac{q^n-1}{q-1}-(q+1)-\frac{q^2-1}{q}\sum_{i=1}^{l-1}q^{4i}-\frac{(q+1)^2}{q^2}\sum_{i=1}^{l}q^{4i}\\
		&=\frac{q^n-1}{q-1}-(q+1)-(q+1)q^2\frac{(q^{4l-3}-q)(q-1)+(q+1)(q^{4l}-1)}{q^4-1}.
	\end{align*}
	This completes the proof.\end{proof}

\begin{example}
Here we give some examples of weight distributions of the second families of FWS codes in Theorem \ref{Th1.2}. Let $n=10,\ l=2$, and $t=5$:
\begin{itemize}
			\item[$(1)$] If $q=2$, then $S=\left<1,\lambda \right>_{\mathbb{F}_{2^2}}\oplus \lambda^2 \mathbb{F}_2$ is a dimension $5$ vector space over $\mathbb{F}_2$, where $\mathbb{F}_2(\lambda)=\mathbb{F}_{2^{10}}$, and the weight distribution of $\mathcal{C}=\text{Orb}(S)$ is:
$$\omega_0=1 ,~\omega_2=2, ~ \omega_4=36, ~ \omega_6=24 ,~ \omega_8 = 576 ,~\omega_{10}=384.$$

			\item[$(2)$] If $q=3$, then $S=\left<1,\lambda \right>_{\mathbb{F}_{3^2}}\oplus \lambda^2 \mathbb{F}_3$ is a dimension $5$ vector space over $\mathbb{F}_3$, where $\mathbb{F}_3(\lambda)=\mathbb{F}_{3^{10}}$, and the weight distribution of $\mathcal{C}=\text{Orb}(S)$ is:
$$\omega_0=1 ,~ \omega_2=3, ~ \omega_4=144, ~ \omega_6=216, ~ \omega_8 = 11664 ,~ \omega_{10}=17496.$$

			\item[$(3)$] If $q=5$, then $S=\left<1,\lambda \right>_{\mathbb{F}_{5^2}}\oplus \lambda^2 \mathbb{F}_5$ is a dimension $5$ vector space over $\mathbb{F}_5$, where $\mathbb{F}_5(\lambda)=\mathbb{F}_{5^{10}}$, and the weight distribution of $\mathcal{C}=\text{Orb}(S)$ is:
$$\omega_0=1,~ \omega_2=5 ,~ \omega_4=900 ,~ \omega_6=3000,~ \omega_8 = 562500, ~ \omega_{10}=1875000.$$
	\end{itemize}
\end{example}

\section{Equivalence classes}\label{sec4}
\noindent In \cite{jcta}, the authors introduced the definition of (linear) isometry between one-orbit cyclic subspace codes and presented some necessary conditions for two (linearly) isometric one-orbit cyclic subspace codes. Besides, they raised the second open problem about the equivalence classes of two families of codes mentioned in Theorem \ref{Th1.2}. In this section, we exhibit some equivalence classes of two families of codes in Theorem \ref{Th1.2} under the condition that the map belongs to the normalizer of $\mathbb{F}_{q^n}^*$.

\begin{definition}\cite[Definition 2.1]{jcta}
	Let \( \mathcal{C}_1, \mathcal{C}_2 \subseteq G_q(n, k) \). \( \mathcal{C}_1 \) and \( \mathcal{C}_2 \) are called (linearly) isometric if there exists an isomorphism \( \psi \in \text{GL}_n(q) \) such that \( \psi(\mathcal{C}_1) = \mathcal{C}_2 \), where
	\[
	\psi(\mathcal{C}_1) = \left\{ \psi(V) : V \in \mathcal{C}_1 \right\}.
	\]
	In this case \( \psi \) is called a (linear) isometry between \( \mathcal{C}_1 \) and \( \mathcal{C}_2 \). In the special case, where \( \mathcal{C}_1 = \text{Orb}(S_1) \), \( \mathcal{C}_2 = \text{Orb}(S_2) \) and \( \psi(\mathcal{C}_1) = \mathcal{C}_2 \) for some \( \psi \in \text{N}_{\text{GL}_n(q)}(\mathbb{F}_{q^n}^*) \), we call the cyclic orbit codes \( \text{Orb}(S_1) \) and \( \text{Orb}(S_2) \) Frobenius-isometric and \( \psi \) a Frobenius isometry.
	Also, if \( \mathcal{C} \subseteq G_q(n, k) \), then the automorphism group of \( \mathcal{C} \) is the group of linear isometries that fix \( \mathcal{C} \), that is \( \text{Aut}(\mathcal{C}) := \left\{ \psi \in \text{GL}_n(q) : \psi(\mathcal{C}) = \mathcal{C} \right\} \).
\end{definition}

Trautmann $\textit{et al}$. in \cite{equivalence} studied the cyclic orbit codes and determined the image of cyclic orbit codes under the (linear) isometry maps.
\begin{lemma}\cite[Theorem 10]{equivalence}\label{...}
	Let $G \leq \mathrm{GL}_n(q)$, $\psi \in \mathrm{GL}_n(q)$, and $U \in \mathcal{G}_q(k,n)$.
	\begin{enumerate}
		\item[$(1)$] Set $G' = \psi G \psi^{-1}$ and $U' = \psi(U)$. Then $\psi(\operatorname{Orb}_G(U)) = \operatorname{Orb}_{G'}(U'), $
		i.e., the two orbit codes are linearly isometric.
		
		\item[$(2)$]  Let $\mathcal{C} = \operatorname{Orb}_G(U)$ and $\mathcal{C}' = \psi(\mathcal{C})$. Then $\mathcal{C}' = \operatorname{Orb}_{\psi G \psi^{-1}}(U') \ \text{with } U' = \psi(U).$
		As a consequence, if $\psi \in \mathrm{N}_{\mathrm{GL}_n}(G)$, then $\mathcal{C}$ and $\mathcal{C}'$ are isometric $G$-orbit codes.
	\end{enumerate}
\end{lemma}
In the following theorem, we determine several classes of equivalence classes of two families of codes in Theorem \ref{Th1.2} under the condition of $\psi \in N_{GL_{n}(\mathbb{F}_{q^n}^*)}$.
\begin{theorem}\label{conclusion}Let $\mathcal{C}$ be a one-orbit cyclic orbit code in $\mathcal{G}_q(k,n)$, and let $\psi\in N_{GL_n(q)(\mathbb{F}_{q^n}^*)}$. Then the followig statements hold:
	\begin{itemize}
		\item[$(1)$]If $\mathcal{C}=\text{Orb}(S)$, then $$\psi(\mathcal{C})=\text{Orb}(S'),$$
		 where $S'=\psi(\mathcal{S})=\alpha^{q^i}\left<1, {\lambda}^{q^i},...,{\lambda^{(q^i)(k-1)}}\right>_{\mathbb{F}_q}$ with $\alpha \in \mathbb{F}_{q^n}^*$ and $i \in \{0,1,..,t-1\}$, where $S$ is of the form in Family $(1)$ of Theorem \ref{Th1.2} and $t=[\mathbb{F}_q(\lambda):\mathbb{F}_q]$. 
		\item[$(2)$]If $\mathcal{C}=\text{Orb}(S)$, then $$\psi(\mathcal{C})=\text{Orb}(S'),$$
		 where $S'=\psi(\mathcal{S})=\alpha^{q^i}\left(\left<1, {\lambda}^{q^i},...,{\lambda^{(q^i)(l-1)}}\right>_{\mathbb{F}_{q^2}}\oplus \lambda^{(q^i)l}{\mathbb{F}_q}\right)$ with $\alpha \in \mathbb{F}_{q^n}^*$ and $i \in \{0,1,..,2t-1\}$, where $S$ is of the form in Family $(2)$ of Theorem \ref{Th1.2} and $t=[\mathbb{F}_{q^2}(\lambda):\mathbb{F}_{q^2}]$.
	\end{itemize}
	\begin{proof}
		\begin{itemize}
			\item[$(1)$]By Lemma $\ref{...}$, $\psi(\mathcal{C})=\text{Orb}_{\psi{\mathbb{F}_{q^n}^*}\psi^{-1}}(\psi(S))=\text{Orb}(\psi(S))$ since $\psi\in N_{GL_{q}}(\mathbb{F}_{q^n}^*)$. Note that the normalizer $N_{GL_{q}}(\mathbb{F}_{q^n}^*)$ of $\mathbb{F}_{q^n}^*$ in $GL_n(q)$ is
			isomorphic to $\text{Gal}(\mathbb{F}_{q^n}|\mathbb{F}_{q})\rtimes \mathbb{F}_{q^n}^* $ by \cite[Theorem 2.4]{Automorphism}. We have $\psi(\mathcal{C})=\alpha^{q^i}\left<1, {\lambda}^{q^i},...,{\lambda^{(q^i)(k-1)}}\right>_{\mathbb{F}_q}$ with $\alpha \in \mathbb{F}_{q^n}^*$ and $i \in \{0,1,..,t-1\}$.  
			\item[$(2)$]Similarly, we have $\psi(\mathcal{C})=\text{Orb}(\psi(S))$. For $\alpha \in \mathbb{F}_{q^n}^*$ and $i \in \{0,1,..,2t-1\}$, assume that $$S''=\alpha^{q^i}\left(\left<1, {\lambda}^{q^i},...,{\lambda^{(q^i)(l-1)}}\right>_{\mathbb{F}_{q^2}}\oplus \lambda^{(q^i)l}{\mathbb{F}_q}\right).$$
			 For any $\beta\in S$, then $\beta=p(\lambda)+a\lambda^l$, where $p(x)\in \mathbb{F}_{q^2}[x]_{\le{l-1}}$ and $a\in \mathbb{F}_q$. One can verify that $\psi(\beta)=\alpha^{q^i}\left({p(\lambda)}^{q^i} + a\lambda^{({q^i})l}\right)$ $\in S''$. Note that $S''$ is a vector space over $\mathbb{F}_q$ and $\text{dim}(\psi(S))_{\mathbb{F}_q}=\text{dim}(S'')_{\mathbb{F}_q}$. Consequently, $\psi(S)=S''$. 
			\qedhere\end{itemize}
	\end{proof}
\end{theorem}
\begin{remark}In particular, $\psi(\mathcal{C})=\mathcal{C}$ if $k=1$ or $i=0$ in case $(1)$. $S=\alpha^{q^i}\mathbb{F}_q$ when $k=1$ and $S=\alpha\langle 1,\lambda,...,\lambda^{k-1}\rangle_{\mathbb{F}_q}$ when $i=0$. Therefore, $\psi(\mathcal{C})=\mathcal{C}$.
\end{remark}

\section{$r$-FWS codes}\label{sec5}
\noindent In \cite{jcta}, the authors introduced the definition of $r$-FWS one-orbit cyclic subspace codes and proposed the problem of determining the classification of $r$-FWS one-orbit cyclic subspace codes. In this section, utilizing the classification results in Lemma $\ref{classification}$ for $r$-FWS one-orbit cyclic subspace codes when $\omega_{2}>0$, we obtain the necessary and sufficient conditions for the existence of $r$-FWS one-orbit cyclic subspace codes which is divided into two classes and present the weight distribution of the second family of $r$-FWS codes. For the first family of codes in Lemma $\ref{classification}$, the result can be directly obtained from Lemma $\ref{thjcta}$. For the second family of codes in Lemma $\ref{classification}$, we prove that no 
$r$-FWS one-orbit cyclic subspace codes exist in most cases and further establish the necessary and sufficient conditions for their existence.

\begin{lemma}\cite[Theorem 5.3]{jcta}\label{classification}
	Let $\mathcal{C}=Orb({S})\subseteq \mathcal{G}(n,k)$, $\omega (\mathcal{C})>0$, $\lambda \in \mathbb{F}_{q^n}\setminus \mathbb{F}_{q}$ such that $d(S,\lambda S)=2$ and $t=[\mathbb{F}_{q}(\lambda):\mathbb{F}_{q}]$. One of the following occurs:
	\begin{itemize}
		\item[$(1)$] if $k<t$ then $S=b\left<1,\lambda,...,\lambda^{k-1}\right>_{\mathbb{F}_q}$, for some $b\in \mathbb{F}_{q^n}^*$;
		\item[$(2)$] if $k\ge t+1$, then $S=\overline S  \oplus b\left<1,\lambda,...,\lambda^{m-1}\right>_{\mathbb{F}_q}$, where $\overline S$ is an $\mathbb{F}_{q^t}$-subspace of dimension $l>0$, $b\in \mathbb{F}_{q^n}^*$, $b{F_{{q^t}}} \cap \overline S  = \{ 0\},\ k = tl + m$ with $0<m<t$.
		
	\end{itemize}
\end{lemma}
In the following lemma, the weight distribution of the first family of FWS codes in Theorem \ref{Th1.2} is given. 
\begin{lemma}\cite[Theorem 4.4]{jcta}\label{thjcta}
	Let \(\lambda \in \mathbb{F}_{q^n} \setminus \mathbb{F}_q\) and let \(t := [\mathbb{F}_q(\lambda) : \mathbb{F}_q]\). Let \(S = \langle 1, \lambda, \ldots, \lambda^{k-1} \rangle_{\mathbb{F}_q}\) with \(1 \leq k < t\). For the code \(\mathcal{C} = \mathrm{Orb}(S)\), we have the following:
	
	If \(k \leq t/2\), then
	\[
	\omega_{2i}(\mathcal{C}) = \begin{cases}
		(q + 1)q^{2i - 1}, & \text{if } i \in \{1, \ldots, k - 1\}, \\
		{} &{}\\
		\displaystyle \frac{q^n - q^{2k - 1}}{q - 1}, & \text{if } i = k.
	\end{cases}
	\]
	
	If \(k > t/2\), then
	\[
	\omega_{2i}(\mathcal{C}) = \begin{cases}
		(q + 1)q^{2i - 1}, & \text{if } i \in \{1, \ldots, t - k - 1\}, \\
		{} &{}\\
		\displaystyle \frac{q^t - q^{2(t - k) - 1}}{q - 1}, & \text{if } i = t - k, \\
		{} &{}\\
		0, & \text{if } i \in \{t - k + 1, \ldots, k - 1\} \text{ and } t - k + 1 \ne k, \\
		{} &{}\\
		\displaystyle \frac{q^n - q^t}{q - 1}, & \text{if } i = k.
	\end{cases}
	\]
\end{lemma}
The following result showed, under certain assumptions, some missing values in the weight distribution of $\mathcal{C}=\text{Orb}(S)$, where $S$ is of the family $2$ in \cite[Theorem 6.1]{jcta}.
\begin{lemma}\cite[Theorem 6.4]{jcta}\label{th6.4}
	Let \((\omega_2(\mathcal {C}), \ldots, \omega_{2k}(\mathcal {C}))\) be the weight distribution of \(\mathcal {C}\). The following hold:
	\begin{enumerate}
		\item[$(1)$] If \(m < t - 1\) and \(H(Y) = \mathbb{F}_{q^t}\), then \(\omega_{2m+2}(\mathcal {C}) = 0\);
		
		\item[$(2)$] If \(m > \frac{t+1}{2}\), then \(\omega_{2(k-j)}(\mathcal {C}) = 0\) for any \(j \in \{1, \ldots, 2m - t - 1\}\);
		
		\item[$(3)$] If \(\mathbb{F}_{q^t} \subsetneq H(Y)\), then \(\omega_{2(k-j)}(\mathcal {C}) = 0\) for any \(j \in \{1, \ldots, 2m - 1\}\);
		
		\item[$(4)$] If \(t = 3\), \(m = 2\), and \(H(Y) = \mathbb{F}_{q^3}\), then \(\omega_4(\mathcal {C}) = 0\).
	\end{enumerate}
\end{lemma}
We will provide some useful facts that will be used in the following theorems.
\begin{lemma}\label{lemcom}
	If \( b \in \mathbb{F}_{q^t}^* \), then
	\[
	S \cap bS = \overline{S} \oplus \left( \langle 1, \lambda, \ldots, \lambda^{m-1} \rangle_{\mathbb{F}_q} \cap b \langle 1, \lambda, \ldots, \lambda^{m-1} \rangle_{\mathbb{F}_q} \right).
	\]
	
	\begin{proof}
		Note that $\overline{S}=b\overline{S}$ with $b\in \mathbb{F}_{q^t}^*$. If \( \mu \in S \cap bS \), then
		\[
		\mu = s_1 + \lambda_1 = s_2 + \lambda_2
		\]
		where \( s_1, s_2 \in \overline{S} \), \( \lambda_1 \in \langle 1, \lambda, \ldots, \lambda^{m-1} \rangle_{\mathbb{F}_q} \), and \( \lambda_2 \in b \langle 1, \lambda, \ldots, \lambda^{m-1} \rangle_{\mathbb{F}_q} \). Therefore,
		\[
		s_1 - s_2 = \lambda_1 - \lambda_2 \in \overline{S} \cap \mathbb{F}_{q^t}.
		\]
		This implies \( s_1 = s_2 \) and \( \lambda_2 = \lambda_1 \). Consequently, \( \mu \) can only be expressed in the form
		\[
		\mu = s + p(\lambda),
		\]
		where \( s \in \overline{S} \), \( p(\lambda) \in \langle 1, \lambda, \ldots, \lambda^{m-1} \rangle_{\mathbb{F}_q} \), and \( p(\lambda) \in b \langle 1, \lambda, \ldots, \lambda^{m-1} \rangle_{\mathbb{F}_q} \).
	\end{proof}
\end{lemma}
\begin{lemma}\label{lem6.3.1}
	If \(\dim_{\mathbb{F}_{q^t}}(\overline{S} \cap \alpha\overline{S}) = i\), then
	\[
	\dim_{\mathbb{F}_q}(S \cap \alpha S) \leq 2m + ti. 
	\]
	
	\begin{proof}
		We know that \(\dim_{\mathbb{F}_q}(\overline{S} + \alpha\overline{S}) = 2tl - ti\). Since \(\overline{S} + \alpha\overline{S} \subseteq S + \alpha S\), it follows that
		\[
		\dim_{\mathbb{F}_q}(S + \alpha S) \geq \dim_{\mathbb{F}_q}(\overline{S} + \alpha\overline{S}) = 2tl - ti.
		\]
		Therefore,
		\[
		\dim_{\mathbb{F}_q}(S \cap \alpha S) = 2k - \dim_{\mathbb{F}_q}(S + \alpha S) \leq 2k - (2tl - ti) = 2m + ti.
		\]
This completes the proof.
	\end{proof}
\end{lemma}
To construct an $r$-FWS code, the intersection dimensions must exhibit highly controlled behavior. In particular, we find that $S$ is either ``close'' to $\overline{S}$ or $Y$ in the sense that $\dim_{\mathbb{F}_q}(\overline{S})\le \dim_{\mathbb{F}_q}(S) \le \dim_{\mathbb{F}_q}(Y)$. Hence the intersection dimension is also ``close'' to the respective intersection dimensions.
Combining \cite[Lemma 6.3 ($4$)]{jcta} and Lemma $\ref{lemcom}$, we derive contradictions that allow us to explicitly identify zero entries in the weight distribution (as done in Lemma $\ref{th6.4}$), which will help us find the bounds of $r$. As a consequence, we show that in most cases, $r$-FWS codes do not exist. 
\begin{theorem}\label{main result3}
	Let $\mathcal{C}$ be a one-orbit cyclic orbit code in $\mathcal{G}_q(n,k)$, and let $\lambda \in \mathbb{F}_{q^n}\setminus \mathbb{F}_{q}$ such that $d(S,\lambda S)=2$ and $t=[\mathbb{F}_{q}(\lambda):\mathbb{F}_{q}]$. Then $\mathcal{C}$ is an $r$-FWS code if and only if $\mathcal{C}=Orb({S})$, where $S$ is one of the following:
	\begin{itemize}
		\item[$(1)$] $S=b\left<1,\lambda,...,\lambda^{k-1}\right>_{\mathbb{F}_q}$ for $\cfrac{t}{2}<k<t=n$ and $r=2k-t$;
		\item[$(2)$] $S=\overline{S}\oplus b\left<1,\lambda,...,\lambda^{m-1}\right>_{\mathbb{F}_q}$ for $t+1\le k \le n$, $r=2m+t(l-1)$, $Y=\mathbb{F}_{q^n}$ and $2m\ge t-1$, where $\overline S$ is an $\mathbb{F}_{q^t}$-subspace of dimension $l>0$, $\mathbb{F}_{q^t}\cap b\overline S=\{0\}$, $k=tl+m$ with $0<m<t$ and $Y=\left<S\right>_{\mathbb{F}_{q^t}}$.
	\end{itemize}
\end{theorem}
\begin{proof}
(1)	The proof of the first part can be seen from Lemma $\ref{thjcta}$. 
	
(2) Let $n_i=\left| \alpha S: \alpha \in \mathbb{F}_{q^n}^*, dim(S\cap \alpha S)=i\right|$ and $S_m=\left<1,\lambda,...,\lambda^{m-1}\right>_{\mathbb{F}_q}$. Then $Y=\left<S\right>_{\mathbb{F}_{q^t}}=\overline{S}+b\left<S_m\right>_{\mathbb{F}_{q^t}}=\overline S \oplus b\mathbb{F}_{q^t}\in \mathcal{G}_{q^t}({n}/{t},l+1)$. Without loss of generality, we assume that $b=1$ in $(2)$. The proof of the second part is divided into four cases:

{\bf Case $1$: $m<t-1$ and $\mathbb{F}_{q^t}=H(Y)$.} By Lemma $\ref{th6.4}$ $(1)$, we have \( n_{k - m - 1} = 0 \), which implies \( r \geq k - m = tl \). There exists \( \gamma \in \mathbb{F}_{q^n} \setminus \mathbb{F}_{q^t} \) such that \( \dim_{\mathbb{F}_{q^t}}(Y \cap \gamma Y) \leq l \) since $H(Y)=\mathbb{F}_{q^t}$ and $t<k\le n$, it follows that
	\[
	\dim_{\mathbb{F}_q}(S \cap \gamma S) \leq \dim_{\mathbb{F}_q}(Y \cap \gamma Y) \leq tl,
	\]
	which means $r=tl$. Otherwise, it would not be an \( r \)-FWS code. Moreover we have
	\[
	tl \leq \dim_{\mathbb{F}_q}(S \cap \gamma S) \leq \dim_{\mathbb{F}_q}(Y \cap \gamma Y) \leq tl,
	\]
	which implies \( \dim_{\mathbb{F}_q}(S \cap \gamma S) = \dim_{\mathbb{F}_q}(Y \cap \gamma Y) = tl \). Furthermore, \(\forall  \gamma \in \mathbb{F}_{q^n} \setminus \mathbb{F}_{q^t},\  S \cap \gamma S = Y \cap \gamma Y \).
	
	In particular, for any nonzero \( \mu \in \overline{S} \), we have \( Y \cap \mu Y = S \cap \mu S\) since \( \overline{S} \cap \mathbb{F}_{q^t} = \{0\} \), and hence we have
	\(
	\mu \mathbb{F}_{q^t} \subseteq Y \cap \mu Y = S \cap \mu S,
	\)
	which forces \( \mu \mathbb{F}_{q^t} \subseteq \mu S \), i.e., \( \mathbb{F}_{q^t} \subseteq S \). Note that \( \lambda^{t-1} \in \mathbb{F}_{q^t} \) and \( \lambda^{t-1} \notin \overline{S} \cup S_m \). Thus, there exists a  polynomial \( p(\lambda) \ne 0 \in S_m \) and \( s_1 \ne 0\in \overline{S} \) such that
	\(
	\lambda^{t-1} = p(\lambda) + s_1.
	\)
	This gives
	\[
	s_1 = \lambda^{t-1} - p(\lambda) \in \overline{S} \cap \mathbb{F}_{q^t},
	\]
	which is a contradiction. Therefore, no such \( r \)-FWS code can exist.
	
{\bf Case $2$: $\mathbb{F}_{q^t}\subsetneq H(Y)$ and $Y\ne \mathbb{F}_{q^n}$.} Note that $\mathbb{F}_{q^t}$ is a subfield of $H(Y)$, which implies that $H(Y)$ is a vector space over $\mathbb{F}_{q^t}$. Let \( s = \dim_{\mathbb{F}_{q^t}}(H(Y)) \geq 2 \) and \( j = \frac{l+1}{s} \). If there exists \( \beta \) such that
	\[
	\dim_{\mathbb{F}_q}(S \cap \beta S) = st \cdot (j - 1) + 1,
	\]
	then \(
	\dim_{\mathbb{F}_q}(Y \cap \beta Y) = t(l + 1).
	\)
	By \cite[Lemma 6.3 ($4$)]{jcta}, we have
	\[
	\begin{aligned}
		st \cdot (j - 1) + 1 &\geq 2m - 2t + \dim_{\mathbb{F}_q}(Y \cap \beta Y) \\
		&\geq 2 - 2t + st \cdot j \\
		&\geq st \cdot (j - 1) + 2,
	\end{aligned}
	\]
	which is impossible. Consequently, \( n_{st(j-1)+1} = 0 \) and \( r \geq st \cdot (j - 1) + 2 \), which implies that for all \( \mu \in \mathbb{F}_{q^n}^* \),
	\[
	\dim_{\mathbb{F}_q}(S \cap \mu S) \geq st \cdot (j - 1) + 2.
	\]
	Since $Y\cap \mu Y$ is a vector space over $\mathbb{F}_{q^t}$, we obtain that
	\(
	\forall \mu \in \mathbb{F}_{q^n}^*, \  \dim_{\mathbb{F}_q}(Y \cap \mu Y) = st \cdot j = t(l + 1),
	\)
	i.e., \( H(Y) = \mathbb{F}_{q^n} \).
	
{\bf Case $3$: $r=2m+t(l-1)$, $2m\ge t-1$ and $Y=\mathbb{F}_{q^n}$.} 

{\bf Case 3.1: \(2m < t\).} Then we have \(\dim_{\mathbb{F}_q}(\overline{S} \cap \beta\overline{S}) \leq t(l - 1)\) for all \(\beta \in \mathbb{F}_{q^n} \setminus H(\overline{S})\) since \(H(\overline{S})\ne \mathbb{F}_{q^n}\). Thus
		\[
		\dim_{\mathbb{F}_q}(\overline{S} + \beta \overline{S}) = 2k - \dim_{\mathbb{F}_q}(\overline{S} \cap \beta \overline{S})\ge 2tl - t(l-1)=n,\]
		which implies $\overline{S}+\beta \overline S=S+\beta S = \mathbb{F}_{q^n}$ and 
		\[
		\dim_{\mathbb{F}_{q^n}}(S \cap \beta S) = 2k - n =  tl + 2m - t .
		\]
		If $\beta \in H(\overline{S})$, then we have
		\[
		\dim_{\mathbb{F}_q}(S\cap \beta S)\ge \dim_{\mathbb{F}_q}(\overline{S} \cap \beta \overline{S}) = tl.
		\]
		If $2m<t-1$ and $r=tl+2m-t$, then we have $n_{tl-1}=0$ and $n_{tl+2m-t} > 0$, in which case there does not exist $r$-FWS code.
		If $2m=t-1$ and $r=tl+2m-t$, then by Lemma $\ref{lemcom}$, for \(i \in \{0,1,\dots,m\}\) we have
		\[
		\dim_{\mathbb{F}_{q^n}}(S \cap \lambda^i S) = tl+m-i,
		\]
		which implies $n_{tl} ,n_{tl+1},\dots,n_{k}$ are all positive.
		
{\bf Case 3.2:  \(2m = t\).} By \cite[Lemma 6.3 ($4$)]{jcta}, we obtain that
		\begin{equation}\label{eq}
			\forall \beta \in \mathbb{F}_{q^n}^*, \quad \dim_{\mathbb{F}_q}(S \cap \beta S) \geq 2m - 2t + t(l + 1) = tl. 
		\end{equation}
		
		By Lemma $\ref{lemcom}$, we have  \[
		\dim_{\mathbb{F}_q}(S \cap \lambda^i S) = tl + i,
		\] which implies \(n_{tl+i} > 0\) for \(i \in \{0,1, \ldots, \frac{t}{2}\}\). Since \(H(\overline{S}) \ne \mathbb{F}_{q^n}^*\), there exists \(\gamma \notin H(\overline{S})\) such that
		\[
		\dim_{\mathbb{F}_q}(\overline{S} \cap \gamma\overline{S}) \leq t(l - 1).
		\]
		By Lemma $\ref{lem6.3.1}$, we have
		\(
		\dim_{\mathbb{F}_q}(S \cap \gamma S) \leq tl.
		\)
		Combining Eq. $(\ref{eq})$, we have \(\dim_{\mathbb{F}_q}(S \cap \gamma S) = tl\). Hence, \(n_{tl} > 0\), confirming that this is a \(tl\)-FWS code.
		
{\bf Case 3.3:  \(2m > t\).} By \cite[Lemma 6.3 ($4$)]{jcta}, we obtain that
		\[
		\dim_{\mathbb{F}_q}(S \cap \beta S) \geq 2m - 2t + t(l + 1) = tl = t(l - 1) + 2m.
		\]
		
		One can verify that
		\[
		\dim_{\mathbb{F}_q}(S \cap \lambda^i S) = tl + m - i
		\]for \(S \cap \lambda^i S\) with \(i \in \{0,1, 2, \ldots, t - m\}\) and \(\beta \in \mathbb{F}_{q^t}\) by Lemma $\ref{lemcom}$.
		Consequently, $n_{tl+i}$ with $i\in\{2m-t, 2m-t+1,...,m\}$ are all positive. This confirms a \(2m + t(l-1)\)-FWS code.

{\bf Case $4$: $m=t-1$ and $H(Y)=\mathbb{F}_{q^t}$.} If $t=2$ and $m=1$, this has been proved to be a $0$-FWS code in \cite{jcta}. If $t=3$ and $m=2$, we have $n_{k-2}=0$ by Lemma $\ref{th6.4}$ $(4)$, which implies $r\ge k-1$.
	According to \cite[Lemma 6.3 ($4$)]{jcta}, we have
$$\dim_{\mathbb{F}_q}(S \cap \mu S) \geq 2m - 2t + \dim_{\mathbb{F}_q}(Y \cap \mu Y) = \dim_{\mathbb{F}_q}(Y \cap \mu Y) - 2 $$
	for any $\mu \in \mathbb{F}_{q^n}^{\ast}$.
	If there exists $\beta \neq 0$ such that $\dim_{\mathbb{F}_q}(S \cap \beta S) = k - 2$ and $t > 3$, then
$\dim_{\mathbb{F}_q}(Y \cap \beta Y) = t(l + 1).$
	Thus, we have
	$$k - 2 \geq t(l + 1) - 2 = k - 1.$$
	
	This is impossible. Therefore, $n_{k-2} = 0$ and $r \ge k - 1$.
	
	If $t\ge 3$, we have shown that for any $\mu \in \mathbb{F}_{q^n}^{\ast}$, $\dim_{\mathbb{F}_q}(S \cap \mu S) \ge r \ge k - 1$. Since $H(Y) = \mathbb{F}_{q^t}$, there exists $\beta \notin H(Y)$ such that
$$\dim_{\mathbb{F}_q}(Y \cap \beta Y) \leq tl.$$
	
	Moreover, we have
	\begin{equation}
		\dim_{\mathbb{F}_q}(S \cap \beta S) \leq \dim_{\mathbb{F}_q}(Y \cap \beta Y) \leq tl.
	\end{equation}

This completes the proof.
\end{proof}

\begin{remark}
	Indeed, the $(tl-1)$-FWS codes and the $tl$-FWS codes are the same as the $(2m + t(l-1))$-FWS codes when $2m=t-1$ and $2m=t$, respectively.
\end{remark}
The following lemma presents a necessary condition for the existence of the second family of $r$-FWS codes in Theorem $\ref{main result3}$, which serves as a valuable tool for determining the weight distribution of such $r$-FWS codes:
\begin{lemma}
		If there exists an $r$-FWS code, then $H(\overline{S}) = \mathbb{F}_{q^t}$.
\end{lemma}
\begin{proof}For any $\beta \in \mathbb{F}_{q^n}$, we have $\beta=s_1+s_2$ with $s_1 \in \overline{S}$, $s_2 \in \mathbb{F}_{q^t}$ since $\mathbb{F}_{q^n}=Y=\overline{S}\oplus \mathbb{F}_{q^t}$ and \[\overline{S} = \left< \alpha, \alpha, \ldots, \alpha^l \right>_{\mathbb{F}_{q^t}} = \alpha \left< 1, \alpha, \ldots, \alpha^{l-1} \right>_{\mathbb{F}_{q^t}},\] where $\alpha$ is the element such that $\mathbb{F}_{q^t}(\alpha) = \mathbb{F}_{q^n}$.
		If there exists $\beta=s_1+s_2 \in \mathbb{F}_{q^n} \setminus \mathbb{F}_{q^t}$ such that $\beta \in H(\overline{S})$, then
	\[
(s_1 + s_2)\overline{S} = \overline{S},
	\]
	where $s_1=a_1\alpha +\cdots+a_i\alpha^i \ne 0,s_2\ne0$ with $a_i\ne 0$ and $i\in \{1,\cdots,l\}$.
Furthermore, we have
		\[
		(s_1 + s_2)\left< 1, \alpha, \ldots, \alpha^{l-1} \right>_{\mathbb{F}_{q^t}} = \left< 1, \alpha, \ldots, \alpha^{l-1} \right>_{\mathbb{F}_{q^t}},\]
which implies $\alpha^{l-i}(s_1 + s_2) \in \left< 1, \alpha, \ldots, \alpha^{l-1} \right>_{\mathbb{F}_{q^t}}$. However, the former contains $\alpha^l$ while the latter does not. Therefore, this is a contradiction and $s_1 = 0$. We have proved that $\beta = s_2 \in \mathbb{F}_{q^t}$.
\end{proof}

In fact, the second family of $r$-FWS codes outlined in Theorem $\ref{main result3}$ can be classified into three distinct cases based on the relationship between $r, t, l$ and $m$:
	\begin{itemize}
		\item[$(1)$]$r=tl-1$ when $2m<t$,
		\item[$(2)$]$r=tl$ when $2m=t$, and
		\item[$(3)$]$r=t(l-1)+2m$ when $2m>t$.
	\end{itemize}
	By analyzing these three cases, we derive the following theorem regarding the weight distribution of such $r$-FWS codes:
\begin{theorem}\label{sta}If $\mathcal{C}=\text{Orb}(S)$ is an $r$-FWS code and $\mu \notin H(\overline{S})$, where $S$ is under the condition of Family $(2)$ in Theorem $\ref{main result3}$, then $\text{dim}_{\mathbb{F}_q}(S\cap \mu S)=t(l-1)+2m$.
\end{theorem}
\begin{proof}If $\mu \notin H(\overline{S}) = \mathbb{F}_{q^t}$, then we have 
		$
		\dim_{\mathbb{F}_q}(\overline{S} \cap \mu \overline{S}) \leq t(l - 1)$
		and
		\[
		t(l - 1) + 2m \geq \dim_{\mathbb{F}_q}(\overline{S} \cap \mu \overline{S}) + 2m \geq \dim_{\mathbb{F}_q}(S \cap \mu S) \geq 2m - 2t + \dim_{\mathbb{F}_q}(Y \cap \mu Y) = t(l - 1) + 2m.
		\]
		by \cite[Lemma 6.3 ($4$)]{jcta}. Thus, if $\mu \notin H(\overline{S}) = \mathbb{F}_{q^t}$, we have 
		\[
		\dim_{\mathbb{F}_q}(S \cap \mu S) = t(l - 1) + 2m.
		\]This completes the proof.
\end{proof}

\begin{theorem}\label{thm2}
In the case of $r=tl-1$, the weight distributions of $\mathcal{C}$ are
		\[
		\omega_{2i}(\mathcal{C}) = 
		\begin{cases} 
			(q + 1)q^{2i-1}, & \text{if } i \in \{1, \ldots, m - 1\}, \\
			{}&{}\\
			\cfrac{q^t - q^{2m-1}}{q - 1}, & \text{if } i = m, \\
			{}&{}\\
			\cfrac{q^n - q^t}{q - 1}, & \text{if } i = m + 1.
		\end{cases}
		\]
\end{theorem}
\begin{proof}
If $\mu \in H(\overline{S})$, then we have 
		$
		\dim_{\mathbb{F}_q}(S \cap \mu S) \geq \dim_{\mathbb{F}_q}(\overline{S} \cap \mu \overline{S}) = tl.
		$
		If $\mu \notin H(\overline{S})$, then by Theorem $\ref{sta}$, 
		\[
		\dim_{\mathbb{F}_q}(S \cap \mu S) = t(l - 1) + 2m = tl - 1,
		\]
		which implies $
		\dim_{\mathbb{F}_q}(S \cap \mu S) = tl - 1 \Leftrightarrow \mu \notin H(\overline{S}).
		$
		Consequently, we obtain 
		\[
		\omega_{2m+2} = \frac{1}{q - 1}|\{\mu \in \mathbb{F}_{q^n} \setminus H(\overline{S})\}| = \cfrac{q^n - q^t}{q - 1}.
		\]
		According to Lemma $\ref{lemcom}$, we only need to consider $\dim_{\mathbb{F}_q}(S_m \cap \mu S_m)$, which has been addressed in Lemma $\ref{thjcta}$ and its proof. Thus, we have
		\[
		|\{\mu \in \mathbb{F}_{q^n} \cap H(\overline{S}) : \dim_{\mathbb{F}_q}(S_m \cap \mu S_m) = m - i\}| =
		\begin{cases}
			(q^2 - 1)q^{2i-1} ,& \text{if } i \in \{1, \ldots, m - 1\}, \\
			&\\
			q^t - q^{2m-1} ,& \text{if } i = m.
		\end{cases}
		\]
		Hence we obtain that
		\begin{align*}
			\omega_{2i}(\mathcal{C})&	= \cfrac{1}{q - 1}|\{\mu \in \mathbb{F}_{q^n} \cap H(\overline{S}) : \dim_{\mathbb{F}_q}(S \cap \mu S) = tl + m - i\}| \\
			&=\begin{cases}
				(q + 1)q^{2i-1}, & \text{if } i \in \{1, \ldots, m - 1\}, \\
				&\\
				\cfrac{q^t - q^{2m-1}}{q - 1}, & \text{if } i = m.
			\end{cases}
		\end{align*}
	
		This completes the proof.
\end{proof}

\begin{theorem}	\label{thm3}
In the case of $r = tl$, the weight distribution of $\mathcal{C}$ is
		\[
		\omega_{2i}(\mathcal{C}) = 
		\begin{cases} 
			(q + 1)q^{2i-1}, & \text{if } i \in \{1, \ldots, m - 1\}, \\
			{}&{}\\
			\cfrac{q^n - q^{2m-1}}{q - 1}, & \text{if } i = m.
		\end{cases}
		\]
\end{theorem}
\begin{proof}If $\mu \notin H(\overline{S})$, then by Theorem $\ref{sta}$, we have 
		\[
		\dim_{\mathbb{F}_q}(S \cap \mu S) = t(l - 1) + 2m = tl.
		\]
		For $\mu \in H(\overline{S})$, similar to Theorem $\ref{thm2}$, we have
		\[
		|\{\mu \in \mathbb{F}_{q^n} \cap H(\overline{S}) : \dim_{\mathbb{F}_q}(S_m \cap \mu S_m) = m - i\}| =
		\begin{cases}
			(q^2 - 1)q^{2i-1}, & \text{if } i \in \{1, \ldots, m - 1\}, \\
			&\\
			q^t - q^{2m-1}, & \text{if } i = m.
		\end{cases}
		\]
		Hence,
		\[
		\omega_{2m} = \cfrac{q^n - q^t}{q - 1} + \cfrac{q^t - q^{2m-1}}{q - 1} = \frac{q^n - q^{2m-1}}{q - 1}.
		\]This completes the proof.
\end{proof}
\begin{theorem}\label{thm4} In the case of $r = 2m + t(l - 1)$, the weight distribution of $\mathcal{C}$ is
		\[
		\omega_{2i}(\mathcal{C}) = 
		\begin{cases} 
			(q + 1)q^{2i-1}, & \text{if } i \in \{1, \ldots, t - m - 1\}, \\
			{}&{}\\
			\cfrac{q^n - q^{2(t-m)-1}}{q - 1} ,& \text{if } i = t - m.
		\end{cases}
		\]
\end{theorem}
\begin{proof}If $\mu \notin H(\overline{S})$, then by Theorem $\ref{sta}$, we have 
		\[
		\dim_{\mathbb{F}_q}(S \cap \mu S) = t(l - 1) + 2m.
		\]
		If $\mu \in H(\overline{S})$, we only need to consider $\dim_{\mathbb{F}_q}(S_m \cap \mu S_m)$.
		Similar to Theorem $\ref{thm2}$, we have
		\begin{align*}
			\omega_{2i}(\mathcal{C}) &=\cfrac{1}{q-1}|\{\mu \in \mathbb{F}_{q^n} \cap H(\overline{S}) : \dim_{\mathbb{F}_q}(S_m \cap \mu S_m) = m - i\}|\\
			&= 
			\begin{cases}
				(q + 1)q^{2i-1} ,& \text{if } i \in \{1, \ldots, t - m - 1\}, \\
				&\\
				\cfrac{q^t - q^{2(t-m)-1}}{q - 1} ,& \text{if } i = t - m.
			\end{cases}
		\end{align*}
		
		Hence,
		\[
		\omega_{2(t-m)} = \cfrac{q^n - q^t}{q - 1} + \cfrac{q^t - q^{2(t-m)-1}}{q - 1} = \cfrac{q^n - q^{2(t-m)-1}}{q - 1}.
		\]This completes the proof.
\end{proof}
The following examples present the weight distributions of the second family of $r$-FWS codes in Theorem $\ref{main result3}$:
\begin{example}Let $q=2,\ n=10,\ t=5$ and $m=2$, then $r=4$ and $l=1$, then $S=\overline{S}\oplus\left<1,\lambda\right>$ is a dimension $7$ vector space over $\mathbb{F}_{2}$, where $\mathbb{F}_2(\lambda)=\mathbb{F}_{2^5}$ and $\overline{S}$ is a dimension $1$ vector space over $\mathbb{F}_{2^5}.$ Since $r=tl-1$, according to Theorem \ref{thm2}, the weight distribution of $\mathcal{C}=\text{Orb}(S)$ is as follows:
$$\omega_0=1,~ \omega_2=6 ,~ \omega_4=24,~ \omega_6=992.$$
\end{example}

\begin{example}Let $q=2,\ n=16,\ t=4$ and $m=2$, then $r=12$ and $l=3$, then $S=\overline{S}\oplus\left<1,\lambda\right>$ is a dimension $14$ vector space over $\mathbb{F}_{2}$, where $\mathbb{F}_2(\lambda)=\mathbb{F}_{2^4}$ and $\overline{S}$ is a dimension $3$ vector space over $\mathbb{F}_{2^4}.$ Since $r=tl$, according to Theorem \ref{thm3}, the weight distribution of $\mathcal{C}=\text{Orb}(S)$ is as follows:
		\begin{equation*}
			\omega_0=1,~ \omega_2=6 ,~\omega_4=65528.
	\end{equation*}
\end{example}
\begin{example}Let $q=3,\ n=9,\ t=3$ and $m=2$, then $r=7$ and $l=2$, then $S=\overline{S}\oplus\left<1,\lambda\right>$ is a dimension $8$ vector space over $\mathbb{F}_{3}$, where $\mathbb{F}_3(\lambda)=\mathbb{F}_{3^3}$ and $\overline{S}$ is a dimension $2$ vector space over $\mathbb{F}_{3^3}.$ Since $r=t(l-1)+2m$, according to Theorem \ref{thm4}, the weight distribution of $\mathcal{C}=\text{Orb}(S)$ is as follows:
		\begin{equation*}
			\omega_0=1,~ \omega_2=9840.
	\end{equation*}
\end{example}

\section{Conclusion}\label{sec7}
In this paper, we obtained the following three main results:
\begin{itemize}
	\item[$(1)$]We presented the weight distribution of the second family of FWS one-orbit cyclic subspace codes in Theorem $\ref{thdistribution}$ mentioned in Theorem \ref{Th1.2} basing on counting the number of pairs $(p_1(x),p_2(x))$ under some conditions.
	\item[$(2)$]We investigated some equivalence classes of two families of FWS one-orbit cyclic subspace codes in Theorem $\ref{conclusion}$ mentioned in Theorem \ref{Th1.2} under the assumption that the map belongs to the normalizer of $\mathbb{F}_{q^n}^*$ in $GL_n(q)$. 
	\item[$(3)$]We classified the $r$-FWS one-orbit cyclic subspace codes in Theorem $\ref{main result3}$ relying on the classification of one-orbit cyclic subspace codes with minmum weight $2$ in Lemma $\ref{classification}$. We found that there are only two classes of $r$-FWS one-orbit cyclic subspace codes and proved that the non-existence of such codes in most cases.
\end{itemize} 

There is still an open problem that naturally arises in this paper: Is it possible to determine all the equivalence classes of FWS one-orbit cyclic subspace codes under the action of linear isometries in Theorem $\ref{Th1.2}$ and $\ref{main result3}$ ?

\end{document}